\documentclass[a4paper,11pt]{article}
\usepackage{a4wide}
\usepackage{amsmath,amsfonts,amssymb,amsthm}
\usepackage[colorlinks,citecolor=blue]{hyperref}
\usepackage{enumitem}
\usepackage{bm}
\usepackage{microtype}

\newtheorem{theorem}{Theorem}[section]
\newtheorem{lemma}[theorem]{Lemma}

\newtheorem{corollary}[theorem]{Corollary}

\theoremstyle{definition}
\newtheorem{definition}[theorem]{Definition}

\numberwithin{equation}{section}

\def \bC {\mathbb{C}}

\def \bN {\mathbb{N}}

\def \bR {\mathbb{R}}

\def \bZ {\mathbb{Z}}

\def \cA {\mathcal{A}}
\def \cB {\mathcal{B}}

\def \cG {\mathcal{G}}

\def \cL {\mathcal{L}}

\def \sinc {\,{\rm sinc}\,}
\def \Re {\,{\rm Re}\,}
\def \Im {\,{\rm Im}\,}
\def \dist {\,{\rm dist}\,}

\begin{document}
\title{Exponential Approximation of Band-limited Functions from Nonuniform Sampling by Regularization Methods}
\author{
Yunfei Yang \thanks{Department of Mathematics, The Hong Kong University of Science and Technology, Clear Water Bay, Kowloon, Hong Kong, P.R. China. E-mail address: {\it yyangdc@connect.ust.hk}.}
\and
Haizhang Zhang \thanks{School of Mathematics (Zhuhai), Sun Yat-sen University, Zhuhai, P.R. China. Supported in part by Natural
Science Foundation of China under grants 11971490 and 12126610. E-mail address: {\it zhhaizh2@sysu.edu.cn}.}
}
\date{}
\maketitle

\begin{abstract}
Reconstructing a band-limited function from its finite sample data is a fundamental task in signal analysis. A Gaussian regularized Shannon sampling series has been proved to be able to achieve exponential convergence for uniform sampling. Whether such an exponential convergence can also be achieved for nonuniform sampling by regularization methods was unresolved. In this paper, we give an affirmative answer to this question. Specifically, we show that one can recover a band-limited function by Gaussian or hyper-Gaussian regularized nonuniform sampling series with an exponential convergence rate. Our analysis is based on the residue theorem in complex analysis, which is used to represent the truncated error by a contour integral. Several concrete examples of nonuniform sampling with exponential convergence will be presented.


\end{abstract}

\section{Introduction}\label{sec: introduction}

The classical Whittaker-Kotelnikov-Shannon sampling theorem \cite{shannon} plays an important role in signal processing, because it establishes the fundamental relation between a band-limited signal and its samples. It states that, for any function $f\in L^2(\bR) \cap C(\bR)$ which is band-limited to $[-\pi,\pi]$ in the sense that the support of its Fourier transform
\[
\widehat{f}(\xi) = \frac{1}{\sqrt{2\pi}} \int_\bR f(x) e^{-ix \xi} dx
\]
is contained in $[-\pi,\pi]$, we can reconstruct $f$ from its infinite sampling data $\{f(n):n\in \bZ \}$ by the cardinal series:
\begin{equation} \label{cardinal series}
f(x)=\sum_{n=-\infty}^{\infty}f(n)\sinc(x-n)=\sum_{n=-\infty}^{\infty} f(n)\frac{\sin \pi(x-n)}{\pi (x-n)},
\end{equation}
where the series converges absolutely and uniformly on $\bR$. For practical reason, we can only sum up finite samples near the point $x$ to approximate $f(x)$. Thus, one has to consider the truncation of the cardinal series. However, due to the slow decayness of the sinc function, truncating the cardinal series leads to a convergence rate of order $O(N^{-1/2})$. Furthermore, this truncated series is an optimal algorithm for recovering functions band-limited to $[-\pi,\pi]$ in the worst case scenario \cite{optimal}.

In order to achieve a fast convergence rate, one may consider functions band-limited to $[-\sigma,\sigma]$ with $\sigma<\pi$, so that the sampling rate of the uniform samples $\{f(n):n\in \bZ \}$ is strictly larger than the Nyquist sampling rate. In this case, it has been shown that regularized Whittaker-Kotelnikov-Shannon sampling series can achieve exponential convergence rates \cite{hyperGaussian,ChenZhang,Lin,Qian1,Qian2,sinc}. Specifically, these papers considered the convergence rate of the regularized sampling series of the form
\[
\sum_{n=-N}^{N}f(n)\sinc(x-n)G_N(x-n).
\]
When the regularizer $G_N$ is a Gaussian function
\[
G_N(x):=\exp\biggl(-\frac {x^2}{2\delta_{\sigma,N}^2}\biggr),\quad x\in\bR,
\]
with $\delta_{\sigma,N}>0$ depending on both $\sigma$ and $N$, Qian \cite{Qian1,Qian2} first used the Fourier analysis method to establish the convergence rate
\[
O \left(N^{1/2}\exp\left(-\tfrac{\pi-\sigma}{2}N\right)\right).
\]
This rate was improved to
\[
O \left(N^{-1/2}\exp\left(-\tfrac{\pi-\sigma}{2}N\right)\right)
\]
in \cite{Lin,optimal} and \cite{sinc} by Fourier methods and complex analysis methods, respectively. Recently, it was shown in \cite{hyperGaussian} that exponential approximation can also be achieved by using a hyper-Gaussian regularizer
\[
G_N(x)=\exp\biggl(-\frac {x^{2m}}{2\delta_{\sigma,N,m}^2}\biggr),\quad x\in \bR,
\]
where $m\in\bN$ and $\delta_{\sigma,N,m}$ depends on $\sigma,N$ and $m$.

Note that all these results are concerned about the regular and uniform sampling points $\Lambda=\bZ$. Uniform sampling is impractical in real applications. Whether exponential convergence can also be achieved for the reconstruction of a band-limited function from its nonuniform sampling remained an open question. The goal of this paper is to study this question for nonuniform sampling sequences.

Sampling theorems on irregular or nonuniform sampling points have been extensively studied in the literature. For example, Higgins \cite{Higgins1,Higgins2} and Hinsen \cite{Hinsen1,Hinsen} proved that a Shannon type sampling formula still hold true for irregular and nonuniform sampling points $\{\lambda_n\}$ satisfying
\[
|\lambda_n-n|\le L,\ \ n\in\bZ
\]
for certain positive constants $L$. Margolis and Eldar \cite{Margolis} considered nonuniform sampling for periodic band-limited functions. The work of Annaby and Asharabi \cite{Annaby} gave upper bound estimates for truncated Shannon's formula for a class of nonuniform sampling points and band-limited functions with fast decay. However, so far, it remains unknown whether exponential convergence can also be achieved in approximating a general band-limited function from its finite samplings on nonuniform points. In this paper, we give affirmative answer to this problem. More precisely, we consider a general sampling sequence $\Lambda=\{\lambda_n \}_{n\in \bZ}$ and use the complex analysis method to analyze the convergence rate of the regularized sampling series
\[
\cG_Nf(z) := \sum_{n=-N}^{N}f(\lambda_n)\varphi_{\Lambda,n}(z)G_N(z-\lambda_n),
\]
where $\varphi_{\Lambda,n}$ is a canonical product to be defined later and $G_N$ is a Gaussian or hyper-Gaussian function. Through careful analysis, we will show that several kinds of regularized nonuniform sampling series $\cG_Nf$ can achieve exponential convergence.

The rest of the paper is organized as follows. We formulate the approximation error $f(z)-\cG_Nf(z)$ as a contour integral in Section \ref{sec: series}. Error estimates for Gaussian and hyper-Gaussian regularizers $G_N$ are given in Sections \ref{sec: gaussian} and \ref{sec: hypergaussian}, respectively. Finally, we present in Section \ref{examples} several examples of nonuniform sampling sequences such that exponential approximation is achieved.

\section{Regularized nonuniform sampling series}\label{sec: series}

Let us begin with an introduction to the spaces of band-limited functions. For any $\sigma>0$, we denote by $B_\sigma$ the set of all entire functions of exponential type at most $\sigma$. In other words, $\cB_\sigma$ consists of functions $f$ that are analytic in the whole complex plane and satisfy
\[
\limsup_{r\to \infty} \frac{1}{r} \log \left(\max_{|z|=r}|f(z)| \right)\le \sigma.
\]
For $1\le p\le \infty$, the Bernstein space $B_\sigma^p$ is the set of all $f\in B_\sigma$ whose restrictions to the real axis belong to $L^p(\bR)$. The norm of $f\in B_\sigma^p$ is defined to be the $L^p$-norm of its restriction on $\bR$. Functions in the Bernstein spaces $B_\sigma^p$ are band-limited in the sense that they have a Fourier transform with compact support in $[-\sigma,\sigma]$ by the Paley-Wiener theorem \cite{rudin}. Furthermore, the Plancherel-P\'olya theorem implies that
\[
B_\sigma^p \subseteq B_\sigma^q \subseteq B_\sigma^\infty \subseteq C(\bR)
\]
for all $1\le p\le q\le \infty$. It is well-known that a function $f\in L^2(\bR) \cap C(\bR)$ is band-limited to $[-\sigma,\sigma]$ if and only if $f$ is the restriction to $\bR$ of a $B_\sigma^2$ function.

We consider the problem of recovering a band-limited function from its samples on a sampling sequence $\Lambda$. For simplicity, throughout this paper, we will assume that $\Lambda :=\{\lambda_n \}_{n\in \bZ}\subseteq \bR $ satisfies
\[
\lambda_0=0, \quad \lambda_n<\lambda_{n+1},\quad n\in \bZ,
\]
and $\Lambda$ is separated in the sense that
\begin{equation}\label{lower}
\delta_\Lambda:= \inf_{n\in\bZ}|\lambda_{n+1}-\lambda_n|>0.
\end{equation}
By the Weierstrass factorization theorem \cite{entire}, there exist entire functions $\varphi_\Lambda$ whose zeros are exactly those points in $\Lambda$. Since
\[
\sum_{n\neq 0} \frac{1}{|\lambda_n|^2}\le \sum_{n\neq 0} \frac{1}{\delta_\Lambda^2 n^2} <\infty,
\]
one of these entire functions is the canonical product
\[
\varphi_\Lambda(z)=z\prod_{n\neq 0}\left(1-\frac{z}{\lambda_n}\right) e^{z/\lambda_n},\quad z\in\bC.
\]
If we further have $\lim_{r\to \infty} \sum_{0<|\lambda_n|<r} \frac{1}{\lambda_n}<\infty$, then
\begin{equation}\label{generating function}
\varphi_\Lambda(z)=z\lim_{r\to \infty} \prod_{0<|\lambda_n|<r}\left(1-\frac{z}{\lambda_n}\right),\quad z\in\bC
\end{equation}
is another simpler choice. Such an entire function $\varphi_\Lambda$ whose zeros are exactly those points in $\Lambda$ is called a generating function of $\Lambda$.

We now turn to the sampling series to be used for reconstruction of a function $f\in\cB_\sigma$ from its sampling at $\Lambda=\{\lambda_n \}_{n\in \bZ}$. First choose a generating function $\varphi_\Lambda$ of $\Lambda$. Since $\varphi_\Lambda(z)$ is an entire function with zeros exactly at the points $\lambda_n$, $n\in \bZ$, we can define for every $n\in \bZ$
\begin{equation}\label{interpolation function}
\varphi_{\Lambda,n}(z)=\frac{\varphi_\Lambda(z)}{\varphi'_\Lambda(\lambda_n)(z-\lambda_n)}.
\end{equation}
These are entire functions which solve the interpolation problem
\[
\varphi_{\Lambda,n}(\lambda_k)=
\begin{cases}
1 \quad &k=n, \\
0 \quad &k\neq n.
\end{cases}
\]
Following the idea of classical Lagrange interpolation, we define
\begin{equation}\label{series}
(\cA_Nf)(z):=\sum_{n=-N}^{N}f(\lambda_n)\varphi_{\Lambda,n}(z).
\end{equation}
Together with the idea of regularization, we shall study the regularized sampling series
\begin{equation}\label{series with entire}
(\cG_N f)(z) :=\sum_{n=-N}^{N}f(\lambda_n)\varphi_{\Lambda,n}(z)G_N(z-\lambda_n),\quad f\in B_\sigma^\infty,
\end{equation}
where $\sigma<\pi$ and $G_N(z)$ is an entire function with $G_N(0)=1$ that will serve as a regularizer. Note that if $\Lambda=\bZ$, then $\varphi_\Lambda(z)=\sin\pi z$, $\varphi_{\Lambda,n}(z)=\sinc(z-n)$ and the series (\ref{series}) is the cardinal series (\ref{cardinal series}). This series with $G_N$ being a Gaussian function and $\Lambda=\bZ$ was used in \cite{sinc} to approximate a band-limited function $f\in B_\sigma^2$ with $\sigma<\pi$.

Next, we use the complex analysis method to estimate the error $f(z)-(\cG_Nf)(z)$. The key idea is that we can represent the error by a contour integral. More specifically, let $\cL_N$ be the positively oriented rectangle with vertices at $T_N^\pm+iS_N^\pm$, where
\[
\lambda_N<T_N^+<\lambda_{N+1},\quad \lambda_{-N-1}<T_N^-<\lambda_{-N},\quad S_N^+>0,\quad S_N^-<0.
\]
Then by the residue theorem, for $z=x+iy\notin \Lambda$ with $T_N^-<x<T_N^+$ and $S_N^-<y<S_N^+$, we can write the error as
\begin{equation}\label{residue}
f(z)-(\cG_Nf)(z)=\frac{\varphi_\Lambda(z)}{2\pi i}\int_{\cL_N}\frac{f(\zeta)G_N(z-\zeta)}{\varphi_\Lambda(\zeta)(\zeta-z)} d\zeta,
\end{equation}
since $G_N(z)$ is an entire function with $G_N(0)=1$. Now, denote by $I_{hor}^\pm$ the contributions to the last integral coming from the two horizontal parts of $\cL_N$, where $+$ and $-$ refer to the upper and the lower line segment, respectively. Similarly, denote by $I_{ver}^\pm$ the contributions coming from the two vertical parts of $\cL_N$, where $+$ and $-$ refer to the right and the left line segment, respectively. Then
\begin{equation}\label{error}
f(z)-(\cG_Nf)(z)=\frac{\varphi_\Lambda(z)}{2\pi i} (I_{hor}^+(z)+I_{hor}^-(z)+I_{ver}^+(z)+I_{ver}^-(z)).
\end{equation}
We now have the following initial estimate of the error $f(z)-\cG_Nf(z)$.

\begin{lemma}\label{iniitalestimate}
Let $f\in B_\sigma^\infty$, $\varphi_\Lambda$ be a generating function of the sampling sequence $\Lambda$, and $G_N$ be an entire function satisfying $G_N(0)=1$. Then for all $z=x+iy\notin \Lambda$ with $T_N^-<x<T_N^+$ and $S_N^-<y<S_N^+$,
\begin{equation}\label{errorbound1}
|f(z)-\cG_Nf(z)|\le \frac{|\varphi_\Lambda(z)|}{2\pi }(|I_{hor}^+(z)|+|I_{hor}^-(z)|+|I_{ver}^+(z)|+|I_{ver}^-(z)|),
\end{equation}
where
\begin{equation}\label{errorbound2}
\begin{aligned}
|I_{hor}^\pm(z)|\le \frac{\|f\|_\infty}{|S_N^\pm-y|}\frac{e^{-(\pi-\sigma)|S_N^\pm|}}{\displaystyle{\min_{\zeta\in \cL_N}\phi_\Lambda(\zeta)}} \int_{T_N^-}^{T_N^+} |G_N(x-t+iy-iS_N^\pm)| dt,
\end{aligned}
\end{equation}
\begin{equation}\label{errorbound3}
\begin{aligned}
|I_{ver}^\pm(z)|\le \frac{\|f\|_\infty}{|T_N^\pm-x|}\frac1{\displaystyle{\min_{\zeta\in \cL_N}\phi_\Lambda(\zeta)}} \int_{S_N^-}^{S_N^+} e^{-(\pi-\sigma)|s|} |G_N(x-T_N^\pm+iy-is)| ds,
\end{aligned}
\end{equation}
with $\|f\|_\infty$ being the $L^\infty$-norm of $f$ on $\bR$ and
\begin{equation}\label{phi Lambda}
\phi_\Lambda(z)=|\varphi_\Lambda(z)|e^{-\pi |\Im z|},\ \ z\in\bC.
\end{equation}
\end{lemma}

\begin{proof} By the Phragm\'en-Lindel\"of principle \cite{entire},
\[
|f(z)|\le \|f\|_\infty e^{\sigma |\Im z|},\ z\in\bC,\ f\in B^\infty_\sigma.
\]
Note also that
\[
\phi_\Lambda(z)=|\varphi_\Lambda(z)|e^{-\pi|\Im z|}>0\mbox{ for }z\in\cL_N.
\]
We now estimate $I_{hor}^\pm$ and $I_{ver}^\pm$ as follows:
\[
\begin{aligned}
|I_{hor}^\pm(x+iy)|&=\left|\int_{T_N^-}^{T_N^+}\frac{f(t+iS_N^\pm) G_N(x+iy-t-iS_N^\pm)}{\varphi_\Lambda(t+iS_N^\pm) (t+iS_N^\pm-x-iy)} dt \right| \\
&\le \int_{T_N^-}^{T_N^+}\frac{\|f\|_\infty e^{\sigma |S_N^\pm|} |G_N(x-t+iy-iS_N^\pm)|}{\displaystyle{(\min_{\zeta\in \cL_N}\phi_\Lambda(\zeta))} e^{\pi |S_N^\pm|} |t-x+iS_N^\pm-iy|} dt \\
&\le \frac{\|f\|_\infty}{|S_N^\pm-y|}\frac{e^{-(\pi-\sigma)|S_N^\pm|}}{\displaystyle{\min_{\zeta\in \cL_N}\phi_\Lambda(\zeta)}} \int_{T_N^-}^{T_N^+} |G_N(x-t+iy-iS_N^\pm)| dt,
\end{aligned}
\]
\[
\begin{aligned}
|I_{ver}^\pm(x+iy)|&=\left|\int_{S_N^-}^{S_N^+}\frac{f(T_N^\pm+is) G_N(x+iy-T_N^\pm-is)}{\varphi_\Lambda(T_N^\pm+is) (T_N^\pm+is-x-iy)} ds \right| \\
&\le \int_{S_N^-}^{S_N^+}\frac{\|f\|_\infty e^{\sigma |s|} |G_N(x-T_N^\pm+iy-is)|}{\displaystyle{(\min_{\zeta\in \cL_N}\phi_\Lambda(\zeta))}e^{\pi |s|} |T_N^\pm-x+is-iy|} ds \\
&\le \frac{\|f\|_\infty}{\displaystyle{(\min_{\zeta\in \cL_N}\phi_\Lambda(\zeta))}|T_N^\pm-x|} \int_{S_N^-}^{S_N^+} e^{-(\pi-\sigma)|s|} |G_N(x-T_N^\pm+iy-is)| ds.
\end{aligned}
\]
The above two equations prove (\ref{errorbound2}) and (\ref{errorbound3}).
\end{proof}

We shall give further upper bound estimates of the error by estimating the generating function $\varphi_\Lambda$ and specifying the regularization function $G_N$.

\section{Gaussian regularization}\label{sec: gaussian}

When the regularizer $G_N(z)$ is a Gaussian function, we can give explicit upper bounds for $|I_{hor}^\pm(z)|$ and $|I_{ver}^\pm(z)|$ so that we can estimate the approximation error $f(z)-(\cG_Nf)(z)$.

\begin{theorem}\label{Gaussian regularization theorem}
Suppose $\Lambda=\{\lambda_n \}_{n\in \bZ} $ is separated, and for every $\lambda_N<T_N^+<\lambda_{N+1}, \lambda_{-N-1}<T_N^-<\lambda_{-N}$, denote $N_*= \min\{\lambda_{-1}-T_N^-,T_N^+-\lambda_1 \}$. Define the series $(\cG_Nf)(z)$ by (\ref{series with entire}), where
\[
G_N(z)=\exp(-r^2z^2) \mbox{ with }r^2=\frac{\pi-\sigma}{2N_*}.
\]
Then, for $N>1$ and $z=x+iy$ satisfying $\lambda_{-1}<x<\lambda_1$ and $|y|<N_*$, it holds for all $f\in B^\infty_\sigma$ with $0<\sigma <\pi$ that
\begin{equation}\label{errorgaussian}
|f(z)-(\cG_Nf)(z)|\le C_N(y) \frac{\|f\|_\infty |\varphi_\Lambda(z)|}{\pi \displaystyle{\min_{\zeta\in \cL_N}\phi_\Lambda(\zeta)}} e^{-\frac{\pi-\sigma}{2}N_*},
\end{equation}
where $\cL_N$ is the rectangle with vertices at $T_N^\pm+i(y\pm N_*)$ and
\[
C_N(y)= \sqrt{\frac{2\pi}{(\pi-\sigma)N_*}} \cosh((\pi-\sigma)y) + \frac{4}{(\pi-\sigma)N_*} \frac{e^{(\pi-\sigma)y^2/(2N_*)}}{1-(y/N_*)^2} .
\]
\end{theorem}

\begin{proof}
By Lemma \ref{iniitalestimate}, we only need to bound $|I_{hor}^\pm(z)|$ and $|I_{ver}^\pm(z)|$. For every $z=x+iy$ satisfying the condition of the theorem, we choose $S_N^\pm=y\pm N_*$. Then $|S_N^\pm|=N_*\pm y$, $|S_N^\pm-y|=N_*$. By inequality (\ref{errorbound2}),
\[
\begin{aligned}
|I_{hor}^\pm(z)| &\le \frac{\|f\|_\infty}{\displaystyle{(\min_{\zeta\in \cL_N}\phi_\Lambda(\zeta))}|S_N^\pm-y|}e^{-(\pi-\sigma)|S_N^\pm|} \int_{T_N^-}^{T_N^+} |G_N(x-t+iy-iS_N^\pm)| dt \\
&\le \frac{\|f\|_\infty e^{\mp(\pi-\sigma)y}}{\displaystyle{(\min_{\zeta\in \cL_N}\phi_\Lambda(\zeta))}N_*} e^{r^2 N_*^2-(\pi-\sigma)N_*} \int_{-\infty}^{\infty} e^{-r^2(x-t)^2} dt \\
&= \frac{\sqrt{\pi}\|f\|_\infty e^{\mp(\pi-\sigma)y}}{r\displaystyle{(\min_{\zeta\in \cL_N}\phi_\Lambda(\zeta))}N_*} e^{r^2N_*^2-(\pi-\sigma)N_*} \\
&= \sqrt{\frac{2\pi}{(\pi-\sigma)N_*}} \frac{\|f\|_\infty}{\displaystyle{\min_{\zeta\in \cL_N}\phi_\Lambda(\zeta)}}e^{\mp(\pi-\sigma)y}  e^{-\frac{\pi-\sigma}{2}N_*}.
\end{aligned}
\]
By the definition of $N_*$, we have $|T_N^\pm -x|\ge N_*$. By inequality (\ref{errorbound3}),
\[
\begin{aligned}
|I_{ver}^\pm(z)| &\le \frac{\|f\|_\infty}{\displaystyle{(\min_{\zeta\in \cL_N}\phi_\Lambda(\zeta))}|T_N^\pm-x|} \int_{S_N^-}^{S_N^+} e^{-(\pi-\sigma)|s|} |G_N(x-T_N^\pm+iy-is)| ds \\
&\le \frac{\|f\|_\infty}{\displaystyle{(\min_{\zeta\in \cL_N}\phi_\Lambda(\zeta))}N_*} e^{-r^2N_*^2} \int_{y-N_*}^{y+N_*} e^{r^2(y-s)^2-(\pi-\sigma)|s|} ds \\
&= \frac{\|f\|_\infty}{\displaystyle{(\min_{\zeta\in \cL_N}\phi_\Lambda(\zeta))}N_*} e^{-\frac{\pi-\sigma}{2}N_*} \int_{-N_*}^{N_*} e^{r^2s^2-(\pi-\sigma)|s+y|} ds.
\end{aligned}
\]
To estimate the last integral, we use the convexity of parabolas to get
\[
r^2s^2-(\pi-\sigma)|s+y|\le
\begin{cases}
\frac{\pi-\sigma}{2}[y+(1-y/N_*)s],\quad & -N_*\le s\le -y, \\
-\frac{\pi-\sigma}{2}[y+(1+y/N_*)s],\quad & -y\le s \le N_*.
\end{cases}
\]
Consequently, we obtain
\[
\int_{-N_*}^{-y} e^{r^2s^2-(\pi-\sigma)|s+y|} ds\le \frac{2e^{(\pi-\sigma)y^2/(2N_*)}}{(\pi-\sigma)(1-y/N_*)}
\]
and
\[
\int_{-y}^{N_*} e^{r^2s^2-(\pi-\sigma)|s+y|} ds\le \frac{2e^{(\pi-\sigma)y^2/(2N_*)}}{(\pi-\sigma)(1+y/N_*)}.
\]
Therefore,
\[
|I_{ver}^\pm(z)| \le \frac{4\|f\|_\infty}{(\pi-\sigma)\displaystyle{(\min_{\zeta\in \cL_N}\phi_\Lambda(\zeta))}N_*} \frac{e^{(\pi-\sigma)y^2/(2N_*)}}{1-(y/N_*)^2} e^{-\frac{\pi-\sigma}{2}N_*}.
\]
Combining these inequalities and using equality (\ref{error}), we get the desired bound (\ref{errorgaussian}).
\end{proof}


In order to get a convergence rate of the regularized nonuniform sampling series, $\Lambda$ needs to be ``dense enough'' so that it is an oversampling for $B^\infty_\sigma$. In view of Theorem \ref{Gaussian regularization theorem}, it is natural to pose conditions on the function $\varphi_\Lambda$.

\begin{corollary}\label{Gaussian regularization corollary}
Assume the conditions of Theorem \ref{Gaussian regularization theorem} and let $T^+_N=\frac{\lambda_N+\lambda_{N+1}}{2}$, $T^-_N=\frac{\lambda_{-N}+\lambda_{-N-1}}{2}$. If there exists $0<\delta<\delta_\Lambda/2$ such that
\begin{equation}\label{lower estimate of generating function}
|\varphi_\Lambda(z)| \ge C|z|^{-p}e^{\pi |\Im z|},\quad \mbox{whenever}\ \dist(z,\Lambda):= \inf_{n\in\bZ}|z-\lambda_n|>\delta,
\end{equation}
for some constants $C>0$ and $p\ge 0$. Then, for any $f\in B^\infty_\sigma$ with $\sigma<\pi$ and $\lambda_{-1}<x<\lambda_1$,
\[
|f(x)-(\cG_Nf)(x)| \le \left( \sqrt{\frac{2\pi}{\pi-\sigma}} + \frac{4}{(\pi-\sigma)\sqrt{N_*}} \right) \frac{\|f\|_\infty |\varphi_\Lambda(x)| \widetilde{N}^p}{C \pi \sqrt{N_*} } e^{-\frac{\pi-\sigma}{2}N_*},
\]
where $\widetilde{N}=\sqrt{\max\{|T_N^+|^2, |T_N^-|^2 \} +N_*^2 }$.
\end{corollary}
\begin{proof}
By hypothesis, $\phi_\Lambda(\zeta)\ge C|\widetilde{N}|^{-p}$ for every $\zeta\in \cL_N$. The result is a direct consequence of the bound (\ref{errorgaussian}).
\end{proof}

If we further know the growth of $|\lambda_N|$ and $|\lambda_{-N}|$, then we can estimate $N_*$ and $\widetilde{N}$ to have a more explicit estimate of $|f(x)-\cG_Nf(x)|$. We give some examples in Section \ref{examples}.

The condition (\ref{lower estimate of generating function}) can be seen as a requirement on the density of $\Lambda$. For instance, if $\Lambda=a\bZ$ for some $a>1$, then $\varphi_\Lambda(z)=\sin(\frac{\pi}{a}z)$ satisfying $|\varphi_\Lambda(z)|\le e^{\frac{\pi}{a}|\Im z|}$. Thus, it cannot satisfy condition (\ref{lower estimate of generating function}) for any $p$. Actually, in this case, $\Lambda$ is an under-sampling for $B^\infty_\sigma$ with $\pi>\sigma>\pi/a$, so there is no hope to reconstruct $f\in B^\infty_\sigma$ from its samples on $\Lambda$.

\section{Hyper-Gaussian regularization}\label{sec: hypergaussian}

In this section, we consider the nonuniform sampling series (\ref{series with entire}) with the hyper-Gaussian regularizer $G_N(z)=\exp(-r_m(N)z^{2m})$, where $m>1$ is an integer and $r_m(N)$ will be chosen later. Recently, it has been proved in \cite{hyperGaussian} by Fourier analysis methods that hyper-Gaussian regularized Whittaker-Kotel'nikov-Shannon sampling series are able to recover a band-limited function from its finite uniform oversampling data. We aim at achieving exponential convergence in reconstructing a band-limited function from its non-uniform sample data by complex analysis methods in this section.

We shall first use the Laplace's method \cite{Laplace method} to estimate $|I_{hor}^\pm|$.
\begin{lemma}[Laplace's method]\label{Laplace}
Let $f$ be a twice differentiable real-valued function on a finite interval $[a,b]$. Assume $c\in (a,b)$ is the only maximum point of $f$ on $[a,b]$, $f'(c)=0$, and $f''(c)<0$. Then
\[
\int_{a}^{b}e^{Nf(t)}dt=e^{Nf(c)}\left( \sqrt{\frac{2\pi}{-f''(c)N}}+o\left( \frac{1}{\sqrt{N}} \right) \right), \quad N\to \infty.
\]
\end{lemma}


\begin{lemma}\label{estimate lemma1}
Let $m>1$ be an integer and $h_m(t)=-\Re (t+i)^{2m}$. Then
\[
\int_{0}^{\infty}e^{Nh_m(t)}dt=e^{N (\sin\frac{\pi}{4m-2} )^{1-2m} } \left(\sqrt{\frac{\pi}{m(2m-1)}} \left(\sin\frac{\pi}{4m-2} \right)^{m-\frac{3}{2}} \frac{1}{\sqrt{N}} +o\left( \frac{1}{\sqrt{N}} \right) \right)
\]
as $N\to \infty$. Consequently, there exists a constant $A_m$ such that for all $N>0$,
\[
\int_{0}^{\infty} e^{Nh_m(t)}dt\le \frac{A_m}{\sqrt{N}}e^{N (\sin\frac{\pi}{4m-2} )^{1-2m} }.
\]
\end{lemma}
\begin{proof}
We first find the extrema of $h_m$ on $[0,\infty)$. Since
$$
h_m'(t)=-2m \Re (t+i)^{2m-1},
$$
the critical points of $h_m$ are $t_k\in [0,\infty)$ such that
$$
\arg(t_k+i)=\frac{\pi+ 2k\pi}{4m-2}, \quad 0\le k\le m-1.
$$
We calculate the value of $h_m$ at $t_k$ as
$$
h_m(t_k)=-\left(\sin \frac{\pi+2k\pi}{4m-2} \right)^{-2m} \cos \left(2m \frac{\pi+2k\pi}{4m-2} \right)=(-1)^k \left(\sin \frac{\pi+2k\pi}{4m-2} \right)^{1-2m}.
$$
Observe that $t_0=\cot \frac{\pi}{4m-2} $ is the only maximum point of $h_m$ on $[0,\infty)$ and
$$
h''_m(t_0)=-2m(2m-1)\left(\sin \frac{\pi}{4m-2} \right)^{3-2m}<0.
$$
On the other hand, since
$$
\lim_{t\to +\infty}\frac{h_m(t)}{t}=-\infty,
$$
there exists some $a>t_0$ such that
$$
\int_{a}^{\infty} e^{Nh_m(t)}dt \le \int_{a}^{\infty} e^{-Nt}dt \le \frac{1}{N}.
$$
By lemma \ref{Laplace},
$$
\int_{0}^{a}e^{Nh_m(t)}dt=e^{N (\sin\frac{\pi}{4m-2} )^{1-2m} } \left(\sqrt{\frac{\pi}{m(2m-1)}} \left(\sin\frac{\pi}{4m-2} \right)^{m-\frac{3}{2}} \frac{1}{\sqrt{N}} +o\left( \frac{1}{\sqrt{N}} \right) \right)
$$
as $N\to \infty$. Combining the last two estimates proves the lemma.
\end{proof}

Similar to the Laplace's method, we will use the following lemma to estimate $|I_{ver}^\pm|$.
\begin{lemma}\label{estimate lemma2}
Let $f$ be a continuous real-valued function on a finite interval $[0,b]$. Assume that $f(t)<f(0)$ for $t\in (0,b]$ and $(f(t)-f(0))/t\to -k$ when $t\to 0$ with $k>0$. Then
$$
\int_{0}^{b}e^{Nf(t)}dt=e^{Nf(0)} \left( \frac{1}{kN}+o\left( \frac{1}{N}  \right) \right), \quad N\to +\infty.
$$
\end{lemma}
\begin{proof}
Without loss of generality, we suppose that $f(0)=0$. Then for any $\epsilon>0$, the maximum of $f(t)$ when $t\ge \epsilon$ is negative. Thus we can write
$$
\int_{0}^{b}e^{Nf(t)}dt=\int_{0}^{\epsilon}e^{Nf(t)}dt + O(e^{-D_\epsilon N})
$$
for some $D_\epsilon>0$ depending on $\epsilon$. Now, we expand $f(t)$ around $t=0$:
$$
f(t)=-kt+o(t).
$$
For any two real numbers $t_1$,$t_2$ we have the inequality
$$
|e^{t_1}-e^{t_2}|\le |t_1-t_2|e^{t_3}, \quad t_3=\max(t_1,t_2).
$$
Applying this with $t_1=Nf(t)$, $t_2=-Nkt$ for $t\in [0,\epsilon]$, we can take $t_3\le -Nkt/2$ by taking $\epsilon$ small enough. With this choice of $\epsilon$ we can write
$$
|e^{Nf(t)}-e^{-Nkt}|\le No(t)e^{-Nkt/2}, \quad 0\le t\le  \epsilon.
$$
Then,
$$
\int_{0}^{\epsilon}|e^{Nf(t)}-e^{-Nkt}|dt=o\left( \frac{1}{N}  \right).
$$
Consequntly,
$$
\begin{aligned}
\left| \int_{0}^{\epsilon} e^{Nf(t)}dt-\frac{1}{kN} \right| &= \left| \int_{0}^{\epsilon} e^{Nf(t)} dt-\int_{0}^{\infty}e^{-Nkt}dt \right| \\
& \le \int_{0}^{\epsilon}|e^{Nf(t)}-e^{-Nkt}|dt+ \int_{\epsilon}^{\infty}e^{-Nkt}dt \\
&= o\left( \frac{1}{N}  \right) + \frac{1}{kN}e^{-Nk\epsilon},
\end{aligned}
$$
which proves the lemma.
\end{proof}

Now, we are ready to use similar arguments as those in Theorem \ref{Gaussian regularization theorem} to deduce an estimate in the Hyper-Gaussian regularization case. Note that in the following theorem, we impose an extra assumption $\sup_N |\lambda_N+\lambda_{-N}|<\infty$, which is not needed in Theorem \ref{Gaussian regularization theorem}.

\begin{theorem}\label{Hyper-Gaussian regularization theorem}
Suppose $\Lambda=\{\lambda_n \}_{n\in \bZ} $ is separated and $\sup_N |\lambda_N+\lambda_{-N}|<\infty$, and for every $\lambda_N<T_N^+<\lambda_{N+1}, \lambda_{-N-1}<T_N^-<\lambda_{-N}$, denote $N_*= \min\{\lambda_{-1}-T_N^-,T_N^+-\lambda_1 \}$. Define the series $(\cG_Nf)(z)$ by (\ref{series with entire}), where $G_N(z)=e^{-r_{m,N}z^{2m}}$ with the integer $m>1$, $r_{m,N}=\mu_m N_*^{1-2m}$ and
$$
\mu_m=\frac{2m-1}{2m}(\pi-\sigma)b_m,\quad  b_m=(2m-1 )^{-\frac{1}{2m}} \left( \sin \frac{\pi}{4m-2} \right)^{\frac{2m-1}{2m}},
$$
then for $N>1$ and $z=x+iy$ satisfying $\lambda_{-1}<x<\lambda_1$ and $|y|<b_mN_*$, there exists a constant $C_m$ depending on $m,\sigma$ and $\Lambda$ such that
$$
|f(z)-(\cG_Nf)(z)|\le C_m \frac{\|f\|_\infty |\varphi_\Lambda(z)|e^{(\pi-\sigma)|y|} }{ \displaystyle{(\min_{\zeta\in \cL_N}\phi_\Lambda(\zeta))} \sqrt{N_*}} e^{-\mu_mN_*},
$$
for every $f\in B^\infty_\sigma$, $0<\sigma<\pi$, where $\cL_N$ is the rectangle with vertices at $T_N^\pm+i(y\pm b_mN_*)$.
\end{theorem}
\begin{proof}
For every $z=x+iy$ satisfying the condition of the theorem, we choose $S_N^\pm=y\pm b_mN_*$ so that $|S_N^\pm|=b_mN_*\pm y$, $|S_N^\pm-y|=b_mN_*$. Observing $\Re(at\pm ia)^{2m}=|a|^{2m}\Re(t+i)^{2m}$ for every $a,t\in \bR$, we have
$$
\begin{aligned}
|I_{hor}^\pm(z)|&\le \frac{\|f\|_\infty}{\displaystyle{(\min_{\zeta\in \cL_N}\phi_\Lambda(\zeta))}|S_N^\pm-y|}e^{-(\pi-\sigma)|S_N^\pm|} \int_{T_N^-}^{T_N^+} |G_N(x-t+iy-iS_N^\pm)| dt \\
&= \frac{\|f\|_\infty}{\displaystyle{(\min_{\zeta\in \cL_N}\phi_\Lambda(\zeta))}|S_N^\pm-y|}e^{-(\pi-\sigma)|S_N^\pm|} \int_{T_N^-}^{T_N^+} e^{-r_{m,N} \Re (x-t+iy-iS_N^\pm)^{2m}} dt \\
&\le \frac{2\|f\|_\infty}{\displaystyle{(\min_{\zeta\in \cL_N}\phi_\Lambda(\zeta))}|S_N^\pm-y|}e^{-(\pi-\sigma)|S_N^\pm|} \int_{0}^{\infty} e^{-r_{m,N} \Re (t+iy-iS_N^\pm)^{2m}} dt \\
&= \frac{2\|f\|_\infty}{\displaystyle{(\min_{\zeta\in \cL_N}\phi_\Lambda(\zeta))}}e^{-(\pi-\sigma)|S_N^\pm|} \int_{0}^{\infty} e^{-r_{m,N}|S_N^\pm-y|^{2m} \Re (t+i)^{2m}} dt \\
&= \frac{2\|f\|_\infty e^{\mp (\pi-\sigma)y}}{\displaystyle{(\min_{\zeta\in \cL_N}\phi_\Lambda(\zeta))}}e^{-(\pi-\sigma)b_mN_*} \int_{0}^{\infty} e^{-\mu_mb_m^{2m}N_*\Re (t+i)^{2m}} dt .
\end{aligned}
$$
By Lemma \ref{estimate lemma1}, there exists a constant $A_m$ such that
$$
\begin{aligned}
|I_{hor}^\pm(z)|&\le \frac{2A_m \|f\|_\infty e^{\mp (\pi-\sigma)y}}{\displaystyle{(\min_{\zeta\in \cL_N}\phi_\Lambda(\zeta))} \sqrt{\mu_mb_m^{2m}N_*}}e^{-(\pi-\sigma)b_mN_*+\mu_mb_m^{2m}N_* (\sin\frac{\pi}{4m-2} )^{1-2m}} \\
&=  \frac{2A_m \|f\|_\infty e^{\mp (\pi-\sigma)y}}{\displaystyle{(\min_{\zeta\in \cL_N}\phi_\Lambda(\zeta))} \sqrt{\mu_mb_m^{2m}N_*}} e^{-\mu_mN_*}.
\end{aligned}
$$

On the other hand,
$$
\begin{aligned}
|I_{ver}^\pm(z)| &\le \frac{\|f\|_\infty}{\displaystyle{(\min_{\zeta\in \cL_N}\phi_\Lambda(\zeta))}|T_N^\pm-x|} \int_{S_N^-}^{S_N^+} e^{-(\pi-\sigma)|s|} |G_N(x-T_N^\pm+iy-is)| ds \\
&= \frac{\|f\|_\infty}{\displaystyle{(\min_{\zeta\in \cL_N}\phi_\Lambda(\zeta))}|T_N^\pm-x|} \int_{S_N^-}^{S_N^+} e^{-(\pi-\sigma)|s|-r_{m,N} \Re (x-T_N^\pm+iy-is)^{2m}} ds \\
&\le \frac{2\|f\|_\infty e^{(\pi-\sigma)|y|}}{\displaystyle{(\min_{\zeta\in \cL_N}\phi_\Lambda(\zeta))}|T_N^\pm-x|} \int_{0}^{b_mN_*} e^{-(\pi-\sigma)|s|-r_{m,N} \Re (T_N^\pm-x+is)^{2m}} ds \\
&= \frac{2\|f\|_\infty e^{(\pi-\sigma)|y|}}{\displaystyle{(\min_{\zeta\in \cL_N}\phi_\Lambda(\zeta))}} \int_{0}^{b_mN_*/|T_N^\pm-x|} e^{-(\pi-\sigma)|T_N^\pm-x|s-r_{m,N} |T_N^\pm-x|^{2m} \Re (1+is)^{2m}} ds \\
&\le \frac{2\|f\|_\infty e^{(\pi-\sigma)|y|}}{\displaystyle{(\min_{\zeta\in \cL_N}\phi_\Lambda(\zeta))}} \int_{0}^{b_m} e^{-(\pi-\sigma)N_* s-r_{m,N} |T_N^\pm-x|^{2m} \Re (1+is)^{2m}} ds.
\end{aligned}
$$

By the assumption that $\sup_N |\lambda_N+\lambda_{-N}|<\infty$, we can decompose $|T_N^\pm-x|=N_*+k^\pm_N(x)$ for some bounded $k^\pm_N(x)$. Therefore,
$$
-r_{m,N} |T_N^\pm-x|^{2m}= -\mu_mN_*-2mk^\pm_N(x)+O(N_*^{-1}),\quad N_*\to \infty.
$$
There hence exists a constant $B_m$ such that
$$
|I_{ver}^\pm(z)|\le \frac{2B_m \|f\|_\infty e^{(\pi-\sigma)|y|}}{\displaystyle{(\min_{\zeta\in \cL_N}\phi_\Lambda(\zeta))}} \int_{0}^{b_m} e^{N_*h_m(s)} ds,
$$
where $h_m(s)=-(\pi-\sigma)s-\mu_m\Re(1+is)^{2m}$. If $h_m(s)<h_m(0)=-\mu_m$ for all $s\in (0,b_m]$, then by Lemma \ref{estimate lemma2}, there exists a constant $B'_m$ such that
$$
|I_{ver}^\pm(z)|\le \frac{2B'_m \|f\|_\infty e^{(\pi-\sigma)|y|}}{(\pi-\sigma) \displaystyle{(\min_{\zeta\in \cL_N}\phi_\Lambda(\zeta))} N_*} e^{-\mu_m N_*}.
$$
To prove the theorem, it remains to show that $h_m(s)<h_m(0)$ for $s\in (0,b_m]$.

Let $s=\tan \beta$, $\beta\in [0,\pi/2)$. We calculate that
$$
h_m'(s)=\sigma-\pi-(-1)^m 2m\mu_m \Re(s+i)^{2m-1}=\sigma-\pi+2m\mu_m (\cos \beta)^{1-2m}\sin(2m-1)\beta.
$$
Observe that $h'_m(s)$ is negative around $s=0$ and increases on $[0,\tan \frac{\pi}{4m-2}]$. Since
\begin{equation}\label{bm estimate}
(2m-1)\sin \frac{\pi}{4m-2}\ge 1>\left(1-\sin^2 \frac{\pi}{4m-2} \right)^m= \left(\cos \frac{\pi}{4m-2} \right)^{2m},
\end{equation}
we have $b_m < \tan \frac{\pi}{4m-2}$. So the maximum point of $h_m(s)$ on $[0,b_m]$ is $s=0$ or $s=b_m$. Therefore, we need to show that $h_m(b_m)<h_m(0)=-\mu_m$, which is equivalent to
$$
\Re (1+ib_m)^{2m}+\frac{1}{2m-1}>0.
$$
By the definition of $b_m$, the above equation is equivalent to
$$
\Re\left( \left( (2m-1)\sin\frac{\pi}{4m-2} \right)^{\frac{1}{2m}}+i \sin\frac{\pi}{4m-2} \right)^{2m}>-\sin\frac{\pi}{4m-2}.
$$
We consider the function
$$
F_m(t)=\Re\left( t +i \sin\frac{\pi}{4m-2} \right)^{2m}, \qquad t\ge 0
$$
and observe that $F_m(\cos\frac{\pi}{4m-2})=\Re e^{i\frac{2m}{4m-2} \pi} =-\sin\frac{\pi}{4m-2}$ and
$$
\begin{aligned}
F'_m(t)&=2m \Re\left( t +i \sin\frac{\pi}{4m-2} \right)^{2m-1}\\
&=2m \left( \sin\frac{\pi}{4m-2} \right)^{2m-1} \Re(\cot\theta+i)^{2m-1} \\
&=2m \left( \sin\frac{\pi}{4m-2} \right)^{2m-1} \frac{\cos (2m-1)\theta}{(\sin\theta)^{2m-1} },
\end{aligned}
$$
where $\cot\theta=t/\sin(\frac{\pi}{4m-2})$. As a consequence, when $t> \cos(\frac{\pi}{4m-2})$, $F'_m(t)> 0$. Using inequality (\ref{bm estimate}), we have $ F_m( ((2m-1)\sin \frac{\pi}{4m-2})^{1/2m} ) > F_m(\cos\frac{\pi}{4m-2}) $, which is the desired inequality.
\end{proof}

We remark that $\mu_m$ is monotonically decreasing as $m$ increases with $\mu_1=\frac{\pi-\sigma}{2}$, which is the same exponent as that in Theorem \ref{Gaussian regularization theorem}. Thus, judging by the exponential term, the Gaussian regularizer is the best among hyper-Gaussian regularizers.

We also have the following corollary.
\begin{corollary}\label{Hyper-Gaussian regularization corollary}
Assume the conditions of Theorem \ref{Hyper-Gaussian regularization theorem} and let $T^+_N=\frac{\lambda_N+\lambda_{N+1}}{2}$, $T^-_N=\frac{\lambda_{-N}+\lambda_{-N-1}}{2}$. If there exists $0<\delta<\delta_\Lambda/2$ such that
$$|\varphi_\Lambda(z)| \ge C|z|^{-p}e^{\pi |\Im z|}\quad whenever\ \dist(z,\Lambda)>\delta$$
for some constants $C>0$ and $p\ge 0$. Then
$$
\begin{aligned}
|f(x)-(\cG_Nf)(x)| &\le C_m \frac{\|f\|_\infty |\varphi_\Lambda(x)| \tilde{N}^p}{C \sqrt{N_*} } e^{-\mu_m N_* }\\
\end{aligned}
$$
for every $\lambda_{-1}<x<\lambda_1$ and $f\in B^\infty_\sigma$ with $\sigma<\pi$, where $\tilde{N}=\sqrt{\max\{|T_N^+|^2, |T_N^-|^2 \} +N_*^2 }$.
\end{corollary}

\section{Examples}\label{examples}

In this section, we provide several examples of nonuniform sampling and prove that the corresponding regularized sampling sequences achieve exponential convergence in reconstructing a band-limited function from its oversampling data.

\subsection{Uniform sampling sequence}

The fundamental example is $\Lambda=\bZ$. In this case, $\varphi_\Lambda(z)=\sin(\pi z)$ and
$$
(\cG_Nf)(x)=\sum_{n=-N}^{N}f(n)\sinc(x-n)G_N(x-n).
$$
We can choose $T_N^+=-T_N^-=N+1/2$, $N_*=N-1/2$, then $\phi_\Lambda\ge 1/2$ on the the rectangle $\cL_N$. Therefore, if $-1<x<1$ and $G_N(x)=\exp(-\frac{\pi-\sigma}{2N-1} x^2 )$, then
$$
\begin{aligned}
|f(x)-(\cG_Nf)(x)| &\le \left( \sqrt{\frac{2\pi}{\pi-\sigma}} + \frac{4}{(\pi-\sigma)\sqrt{N-\frac{1}{2}}} \right) \frac{2\|f\|_\infty |\sin(\pi x)|}{ \pi \sqrt{N-\frac{1}{2}} } e^{-\frac{\pi-\sigma}{2}\left(N-\frac{1}{2} \right) }\\
&=O(N^{-\frac{1}{2}} e^{-\frac{\pi-\sigma}{2} N} )
\end{aligned}
$$
for every $f\in B^\infty_\sigma$ with $0<\sigma<\pi$. If $G_N(x)=\exp(-\mu_m(N-1/2)^{1-2m}x^{2m} )$, we have
$$
|f(x)-(\cG_Nf)(x)|\le C_m \frac{\|f\|_\infty |\sin(\pi x)|}{\sqrt{N-1/2}} e^{-\mu_m(N-1/2)}=O(N^{-1/2}e^{-\mu_mN}),\ -1<x<1.
$$
The focus of the paper is of course on nonuniform sampling. The above example is to show that our analysis also recovers those results established in \cite{Lin,Qian1,Qian2,sinc} and \cite{hyperGaussian} for uniform sampling.

\subsection{Zeros of a sine-type function}

\begin{definition}[Sine-type functions]\label{sine-type}
An entire function $\varphi$ of exponential type $\pi$ is said to be a sine-type function if it has simple and separated zeros and there exist positive constants $A,B,H$ such that
$$
Ae^{\pi |y|}\le |\varphi(x+iy)|\le Be^{\pi |y|}\ \ \ for\ all\ x\in\bR\ and\ |y|\ge H.
$$
\end{definition}
The zeros of a sine-type function lie in a horizontal strip and if we enumerate them in increasing order of their real parts then $\Lambda$ satisfies (\ref{lower}) and
$$
\sup_{n\in\bZ}|\lambda_{n+1}-\lambda_n|<\infty.
$$
Moreover, for each $\epsilon>0$, there exist constants $M_1$ and $M_2$ such that
\begin{equation}\label{sine-type inequality}
0<M_1<|\varphi(z)|e^{-\pi |\Im z|}<M_2<\infty,\quad\quad dist(z,\Lambda)>\epsilon.
\end{equation}
Any sine-type function $\varphi$ can be determined from its zero set $\Lambda$ by (\ref{generating function}). If $\Lambda\subset \bR$ is the zeros of a sine-type function, then $\Lambda$ is a complete interpolating sequence for $B^2_\pi$. Consequently, it has a uniform density in the sense that for every $x\in \bR$,
$$
D(\Lambda)=\lim_{r\to\infty}\frac{\# (\Lambda\cap [x,x+r])}{r}=1.
$$
Readers are referred to \cite{sampling,Higgins,entire,nonharmonic} for more information on sine-type functions.

Here is a simple way to construct sine-type functions. For any function $g$ that can be represented as
$$
g(z)=\frac{1}{\sqrt{2\pi}}\int_{-\sigma}^{\sigma}h(\xi)e^{i\xi z} d\xi,\ \ \ z\in \bC
$$
for some $h\in L^1[-\sigma,\sigma]$, we can define the sine wave crossings of $g$ as
$$
\varphi_g(z)=A\sin(\pi z)-g(z),\qquad z\in \bC
$$
 where $A$ is a constant such that $A>\|h\|_{L^1}$. These functions are all sine-type functions. Moreover, the zeros of $\varphi_g$ are all real and simple if $g$ is real on the real axis \cite{sine crossing,sine crossing2,sine crossing3}. Therefore, given any function $g\in B^2_\pi$, we can construct corresponding sine-type functions in this way. Note that the zeros of the sine-type function $\varphi(z)=\sin(\pi z)$ corresponding to $g=0$ is the uniform sampling sequence $\Lambda=\bZ$.

Now, suppose that $\Lambda\subset \bR$ is the zeros of a sine-type function $\varphi_\Lambda$, then we know that $\delta_\Lambda:= \inf_{n\in\bZ}|\lambda_{n+1}-\lambda_n|>0$. If we choose a $\epsilon<\delta_\Lambda/2$ in (\ref{sine-type inequality}), and
$$
T_N^+=\frac{\lambda_N+\lambda_{N+1}}{2},\quad T_N^-=\frac{\lambda_{-N}+\lambda_{-N-1}}{2},\quad S_N^+>\epsilon,\quad S_N^-<-\epsilon,
$$
then by (\ref{sine-type inequality}), $\phi_\Lambda>M_1$ on the rectangle $\cL_N$. The density $D(\Lambda)=1$ implies
$$
\lim_{N\to \infty} \frac{N}{|T^\pm_N|}=1.
$$
Thus, for every $0<\eta<1$, $|T^\pm_N|>\eta N$ for sufficiently large $N$. Therefore, for $\lambda_{-1}<x<\lambda_1$ and $f\in B^\infty_\sigma$, if $G_N(x)=\exp(-\frac{\pi-\sigma}{2\eta N} x^2 )$, then
$$
|f(x)-(\cG_Nf)(x)|=O(N^{-\frac{1}{2}} e^{-\frac{\pi-\sigma}{2}\eta N} ),\quad \forall\  0<\eta<1.
$$
If  $\sup_N |\lambda_N+\lambda_{-N}|<\infty$ and $G_N(x)=\exp(-\mu_m(\eta N)^{1-2m}x^{2m} )$, we have
$$
|f(x)-(\cG_Nf)(x)|=O(N^{-\frac{1}{2}}e^{-\mu_m \eta N}),\quad \forall\  0<\eta<1.
$$
As far as we know, these two results with exponential convergence for nonuniform sampling is new.

\subsection{Perturbed uniform sequence}

\begin{definition}[Perturbed uniform sequences]\label{perturbed uniform sequence}
A sequence $\Lambda=\{\lambda_n \}_{n\in \bZ}$ of real number is called a perturbed uniform sequence with $L\ge0$ if
$$
|\lambda_n-n|\le L \mbox{ for all }n\in \bZ.
$$
\end{definition}

Suppose that $\Lambda$ is a perturbed uniform sequence with $L$. If $L< \frac{1}{4}$, the celebrated Kadets $\frac{1}{4}$ theorem \cite{entire} shows that it is a complete interpolating sequence for $B^2_\pi$. However, when $L\ge \frac{1}{4}$,  $\Lambda$ may not be a complete interpolating sequence. Nevertheless, the following lemma proved in \cite{Hinsen} provides an estimate of the generating function $\varphi_\Lambda(z)$ of $\Lambda$ when $L<\frac{1}{2}$.
\begin{lemma}\label{estimate}
Suppose that $\Lambda=\{\lambda_n \}_{n\in \bZ}$ is a perturbed uniform sequence with $L<\frac{1}{2}$ and $\lambda_0=0$. Then the generating function $\varphi_\Lambda(z)$ of $\Lambda$ defined by (\ref{generating function}) is an entire function, and there are constants $C_1$, $C_2$ such that for all $z=|z|e^{i\theta} \in \bC$ with $|z|$ large enough,
$$
C_1H_1(z)H_2(L,z)\le |\varphi_\Lambda(z)|\le C_2H_1(z)H_2(-L,z),
$$
where
\begin{equation*}
H_1(z):=e^{\pi|\Im z|} \left\{
\begin{aligned}
& 1, \qquad && |\Im z|>1, \\
& \prod_{k=N(z)}^{N(z)+2}|\lambda_k-z|,\quad && |\Im z|\le 1\quad and\quad \Re z>0, \\
& \prod_{k=-N(z)-2}^{-N(z)}|\lambda_k-z|,\quad && |\Im z|\le 1\quad and\quad \Re z<0,
\end{aligned}
\right.
\end{equation*}
\begin{equation*}
H_2(d,z):=\left\{
\begin{aligned}
& |z|^{-4d},\quad & 0\le |\sin\theta|\le \sin(\pi/(2|z|)), \\
& |z|^{-2d}|\sin\theta |^{2d},\quad &\sin(\pi/(2|z|))<|\sin\theta |\le 1
\end{aligned}
\right.
\end{equation*}
with $N(z)$ being a positive integer dependent on $z$.
\end{lemma}

If $\Lambda$ is a perturbed uniform sequence with $L<\frac{1}{2}$, then $\delta_\Lambda\ge 1-2L$. For every $\epsilon<\frac{1}{2}-L$, we have $|\lambda_k-z|>\epsilon$ whenever $\dist(z,\Lambda)> \epsilon$. By Lemma \ref{estimate}, there exists a constant $C$ such that
$$
\phi_\Lambda(z)=|\varphi_\Lambda(z)|e^{-\pi|\Im z|}\ge C |z|^{-4L},\  \dist(z,\Lambda)> \epsilon.
$$
Thus, if we choose $T_N^+=-T_N^-=N+\frac{1}{2}$, then $\phi_\Lambda(z)\ge C |z|^{-4L}$ on the rectangle $\cL_N$ and $N-\frac{1}{2}-L\le N_*\le N-\frac{1}{2}+L$. Therefore, if $G_N(x)=\exp(-\frac{\pi-\sigma}{2N-2} x^2 )$, then
$$
|f(x)-(\cG_Nf)(x)|=O(N^{4L-\frac{1}{2}} e^{-\frac{\pi-\sigma}{2} N} ),\ \lambda_{-1}<x<\lambda_1.
$$
If $G_N(x)=\exp(-\mu_m(N-1)^{1-2m}x^{2m} )$, we have
$$
|f(x)-(\cG_Nf)(x)|=O(N^{4L-\frac{1}{2}}e^{-\mu_mN}),\ \lambda_{-1}<x<\lambda_1.
$$
Note that when $L=0$, the estimates reduce to the case $\Lambda=\bZ$. When $L\ne0$, the above two results are also new to our best knowledge.

{\small
\bibliographystyle{amsplain}

}

%
\end{document}